\documentclass[lettersize,journal]{IEEEtran}
\usepackage{amsmath,amsfonts,amsthm,amssymb}
\usepackage{mathrsfs}
\usepackage{dutchcal}
\usepackage{algorithm}
\usepackage{algorithmicx}
\usepackage{algpseudocode}
\usepackage{epsfig}
\usepackage{subfigure}
\floatname{algorithm}{Algorithm}

\theoremstyle{remark}
\newtheorem{theorem}{\hskip 1em Theorem}
\newtheorem{proposition}{\hskip 1em Proposition}
\newtheorem{lemma}{\hskip 1em Lemma}

\newtheorem{assumption}{\hskip 1em Assumption}
\newtheorem{property}{\hskip 1em Property}

\newtheorem{proof skecth}{Proof skecth}
\usepackage{array,caption}
\usepackage[mathscr]{eucal}
\usepackage[caption=false,font=normalsize,labelfont=sf,textfont=sf]{subfig}
\usepackage{textcomp}
\usepackage{stfloats}
\usepackage{url}
\usepackage{verbatim}
\usepackage{graphicx,float,balance}
\usepackage{multirow}
\usepackage{cite}
\usepackage{soul}
\usepackage{color,xcolor}
\def\BibTeX{{\rm B\kern-.05em{\sc i\kern-.025em b}\kern-.08em
		T\kern-.1667em\lower.7ex\hbox{E}\kern-.125emX}}
\usepackage{balance}
\begin{document}	
	\title{Performance Analysis of Perturbation-enhanced SC decoders}
	\author{\IEEEmembership{$\mathrm{Zhicheng\ Liu}$, \and $\mathrm{Liuquan\ Yao}$, \and $\mathrm{Shuai\ Yuan}$, \and $\mathrm{Guiying\ Yan}$, \and $\mathrm{Zhiming\ Ma, \ and}$ \and $\mathrm{Yuting\ Liu}$}
		\thanks{This work was supported by the National Key R\&D Program of China under Grant 2023YFA1009601. (\textit{Corresponding author: Yuting Liu.})}
		\thanks{Zhicheng Liu and Yuting Liu are with the School of Mathematics and Statistics, Beijing Jiaotong University, Beijing, 100044, China (e-mail: 20118002@bjtu.edu.cn, ytliu@bjtu.edu.cn).}%, and also with Huawei Technologies Co. Ltd., Hangzhou 310053, China 
		%\thanks{Yuan Li, Huazi Zhang and Jun Wang are with Huawei Technologies Co. Ltd., Hangzhou 310053, China (e-mail: \{liyuan299, zhanghuazi, justin.wangjun\}@huawei.com).}
		
		%\thanks{Wen Tong is with Huawei Technologies Co. Ltd., Ottawa K2K 3J1, Canada (e-mail: tongwen@huawei.com).}
		
		\thanks{Liuquan Yao, Shuai Yuan, Guiying Yan and Zhiming Ma are with the University of Chinese Academy of Sciences, Academy of Mathematics and Systems Science, CAS, Beijing 100190, China (e-mail: yaoliuquan20@mails.ucas.ac.cn, yuanshuai2020@amss.ac.cn, yangy@amss.ac.cn, mazm@amt.ac.cn).}
		
		%\thanks{The corresponding author is Yuting Liu (email: ytliu@bjtu.edu.cn).}
		}
	
	
	\maketitle
	
	\begin{abstract}
		In this paper, we analyze the delay probability of the first error position in perturbation-enhanced Successive cancellation (SC) decoding for polar codes. Our findings reveal %theoretical analysis demonstrates the effectiveness of the algorithm by examining how a single perturbation affects the first error position. We find 
		that, asymptotically, an SC decoder's performance does not degrade after one perturbation, and it improves with a probability of $\frac{1}{2}$. %Performance degradation is indicated by an earlier first bit error, while improvement is achieved by delaying this error. 
		This analysis explains the sustained performance gains of perturbation-enhanced SC decoding as code length increases.
	\end{abstract}
	
	
	\section{Introduction}
	\IEEEPARstart{P}{olar} codes, introduced by Arıkan \cite{ref1}, are the first capacity-achieving codes and have significantly impacted information and coding theory. In recent years, both academic and industrial interest has grown, with polar codes being adopted for control channels in the 5G eMBB scenario \cite{ref2}, driving further research into decoding algorithms.
	%\IEEEPARstart{P}{olar} codes, introduced by Arıkan \cite{ref1}, have made a significant impact on information and coding theory as the first proven capacity-achieving codes. Over the past decade, interest in polar codes has surged in both academia and industry. Polar coding was adopted for control channels in the 5G eMBB scenario \cite{ref2}, spurring further research into advanced decoding algorithms.
	
	%Successive Cancellation (SC) decoding, pioneered by Arıkan \cite{ref1}, is known for its simplicity. For finite code lengths, performance can be further improved with enhanced methods like Log-Likelihood Ratio-based Successive Cancellation List (LLR-based SCL) and Cyclic Redundancy Check-aided SCL (CA-SCL) algorithms \cite{ref3}, \cite{ref4}, \cite{ref5}.
%	\IEEEPARstart{P}{olar} codes, introduced by Arıkan \cite{ref1}, have significant impact on information and coding theory since it is the first proven capacity-achieving codes. Over the past decade, both academic and industrial interest in polar codes has surged.	Polar coding was adopted for control channels in the 5G eMBB scenario \cite{ref2}, spurring further research into advanced decoding algorithms.
%	
    Successive cancellation (SC) decoding, pioneered by Arikan \cite{ref1}, enjoys extreme simplicity. For finite code lengths, enhanced decoding methods, such as log-likelihood ratio-based successive cancellation list (LLR-based SCL) and cyclic redundancy check-aided SCL (CA-SCL) algorithms \cite{ref3}, \cite{ref4}, \cite{ref5} can further improve performance. %Despite their comparable computational complexity with SC decoding in terms of block length, SCL algorithms suffer from high decoding latency due to their sequential nature, making them ill-suited for iterative detection and decoding as they produce hard outputs.

	Other performance-enhancing methods, such as flipping \cite{ref6} and automorphism ensemble decoding \cite{ref7}, \cite{ref8}, \cite{ref9}, show potential but face challenges. The efficiency of SCL-flipping decoding decreases with code length $N$ as unreliable bits grow \cite{ref10}. Additionally, Urbanke et al. \cite{ref11} show that the useful automorphism group size (SC-variant permutations) decreases asymptotically for polar codes constructed using Arikan’s method, reducing their efficiency for longer codes.  %Additionally, Urbanke et al. \cite{ref12} reveal that the size of the useful automorphism group decreases asymptotically for polar codes constructed using Arikan’s method, reducing their efficiency for longer codes.
	The $y$-side perturbation-enhanced decoding method, introduced in \cite{ref12} and shown to improve SCL decoders with increasing code length \cite{ref13}, enhances SC/SCL decoders by adding small perturbations to the received LLRs upon CRC failure. This simple and effective method yields noticeable performance gains across various code lengths and rates \cite{ref13}. However, the reasons for these improvements are unclear. To simplify analysis, we use SC decoders as a reference.

	In this paper, we offer theoretical insights into perturbation-enhanced SC decoding with the following contributions:
	\begin{enumerate}
		\item We introduce a $u$-side (hard-decision side) perturbation-enhanced SC decoder, where perturbations are directly applied to the decision LLRs (the LLRs used for hard decision) during the SC decoding process.
		\item We prove that, under reasonable assumptions, the first error position (the first bit where an SC decoder fails) will asymptotically either remain the same or be delayed with a probability of $\frac{1}{2}$, implying that perturbation preserves or enhances the performance of SC decoder.
	\end{enumerate}
	%	In this paper, we provide some theoretical insights into SC+perturbation decoding with the following contributions:
	%	\begin{enumerate}
		%		\item We design a $u$-side SC+perturbation decoder. %studied the impact of perturbation on the $u$-side LLRs. %values.
		%		\item We prove that, asymptotically, both the delay probability and the unchanged probability of the first error position converge to $\frac{1}{2}$.
		%	\end{enumerate}
	
	The remainder of this paper is organized as follows: Section II covers the preliminaries of polar codes and SC/SCL-based decoding algorithms. The delay probability of the first error position is derived in Section III. Section IV presents some simulation results. Finally, Section V concludes the paper.
	%	In this paper, we give some theoretical explanations for SC+perturbation decoding with the following contributions:
	%	\begin{enumerate}
		%		\item{We design a $u$-side SC+perturbation decoder.}
		%		\item{We prove that in the asymptotic sense, the delay probability of the first error position is $\frac{1}{2}$ and the first error position remains unchanged with a probability of $\frac{1}{2}$.}
		%	\end{enumerate}
	%	
	%	The remainder of this paper is organized as follows: Section II describes some preliminaries of polar codes and SC/SCL-based decoding algorithms. The delay probability of the first error position is derived in Section III. Then some simulation results are presented in Section IV. Finally, Section V concludes this paper.
	
	\textit{Notation Conventions:} This paper defines the probability distribution of $X$ via its probability density function (PDF) $P_X(x)$. Here, $\mathbb{P}(\cdot)$ denotes probability, $\mathbb{E}(\cdot)$ is expectation, $\mathbb{I}_A$ is the indicator function for set $A$ and $\mathrm{sgn}(\cdot)$ is the sign function. The complementary cumulative distribution function (CDF) of the standard normal distribution is $\mathrm{Q}(x)=\int_x^{+\infty} e^{-\frac{t^2}{2}} \mathrm{~d} t$. The partial derivative of $f(x, y)$ with respect to $x$ is $\frac{\partial f(x, y)}{\partial x}$, and boldface denotes matrices and vectors. %The partial derivative of $f(x, y)$ with respect to $x$ is denoted as $\frac{\partial f(x, y)}{\partial x}$, and boldface is used for all matrices and vectors.
	\section{Preliminaries}
	\subsection{Construction of Polar Codes}
	Polar codes of length $N=2^{n}$ are constructed as follows:
	\begin{flalign*}
		\mathbf{c}_{0:N}=\mathbf{u}_{0:N}\mathbf{G}_{N},%\label{eq1}
	\end{flalign*}
	where $\mathbf{c}_{0:N}=[c_{0},\cdots,c_{N-1}]$ are the encoded codewords, $\mathbf{u}_{0:N}=[u_{0},\cdots,u_{N-1}]$ are the source bits, and \( \mathbf{G}_{N} = \mathbf{B}_{n} \mathbf{F}^{\otimes n} \), where $\mathbf{F}^{\otimes n}$ denotes the $n$-th Kronecker product of \( \begin{bmatrix}
		1 & 0 \\
		1 & 1
	\end{bmatrix} \), and \( \mathbf{B}_{n} \) is the bit-reversal permutation matrix \cite{ref14}.
	
	Channel polarization, achieved through recursive polarization transformations, creates either noiseless or noisy channels. For encoding, the $K$ most reliable indices in $[0,\cdots,N-1]$ are used for information bits (denoted as $\mathscr{A}$), while the remaining $N-K$ indices are assigned all-zero frozen bits.
	
	%The encoded bits $\mathbf{c}_{0:N}$ are modulated using BPSK as \(x_{i} = 1 - 2c_{i}\) and transmitted over AWGN channels. The received symbols are \(y_{i} = x_{i} + n_{i}\), where \(n_{i} \sim \mathscr{N}(0, \sigma^2)\) is Gaussian noise with variance \(\sigma^2\).
	%Channel polarization, achieved through recursive transformations, produces either completely noiseless or noisy channels. For encoding, the $K$ most reliable channel indices in $[0,\cdots,N-1]$ are used for information bits (denoted as $\mathscr{A}$), while the remaining $N-K$ channels are assigned frozen bits.
	
	The encoded bits $\mathbf{c}_{0:N}$ are modulated using binary phase shift keying (BPSK) as \(x_{i} = 1 - 2c_{i}\) and transmitted over  additive white Gaussian noise (AWGN) channels. The received symbols are given by \(y_{i} = x_{i} + n_{i}\) and $\mathbf{y}_{0:N}=[y_{0},\cdots,y_{N-1}]$, where \(n_{i} \sim \mathscr{N}(0, \sigma^2)\) represents noise with variance \(\sigma^2\).
	%Channel polarization, achieved through recursive polarization transformations, creates either completely noiseless or noisy channels. The encoding strategy uses the $K$ most reliable bit-channels for information bits and assigns frozen bits to the remaining $N-K$ noisy channels.
	%
	%The code bits $\{c_{i}\}_{i=1}^{N}$ are modulated using binary phase shift keying (BPSK) as $x_{i} = 1 - 2c_{i}$, and transmitted through additive white Gaussian noise (AWGN) channels. The received symbols are represented as $y_{i} = x_{i} + n_{i}$, where $n_{i} \sim \mathscr{N}(0, \sigma^{2})$ and $\sigma^{2}$ is the noise power.
	%Channel polarization is achieved by recursively applying this polarization transformation, resulting in a set of completely noiseless or completely noisy channels. Therefore, a natural encoding strategy is to transmit free bits (called information bits) over the $K$ most reliable bit-channels, denoted by $\mathscr{A}$, while assigning fixed bits (known as frozen bits) to the remaining $N-K$ noisy channels.
	%
	%The code bits $\{c_{i}\}_{i=1}^{N}$ are modulated by binary phase shift keying (BPSK), i.e., $x_{i} = 1 - 2c_{i}$, and then transmitted through additive white Gaussian noise (AWGN) channels. The received symbols can be represented as $y_{i} = x_{i} + n_{i}$, where $n_{i} \sim \mathscr{N}(0, \sigma^{2})$ and $\sigma^{2}$ denotes the noise power.
	\subsection{SC/SCL decoding of Polar Codes}
	Denote $ L^{(i)}_{1} = \frac{2y_i}{\sigma^2} $ the received LLR value of bit $i$. For a length-$2$ polar code $\mathbf{u}_{0:2}$, SC decoding proceeds as follows:
	
	The $f$-operation calculates the LLR value of bit $u_0$ as:
	\begin{flalign*}
		L(u_0)=f(L^{(0)}_{1},L_{1}^{(1)})=\mathrm{sgn}(L_{1}^{(0)}, L_{1}^{(1)}) \min\{|L_{1}^{(0)}|, |L_{1}^{(1)}|\}.%\label{eq2}
	\end{flalign*}
	
	Then, a hard decision on $ L(u_0) $ yields the estimate $\hat{u}_0 $.
	
	The $ g $-operation updates the LLR value of $ u_1 $:
	\begin{flalign*}
		L(u_1) = g(L_{1}^{(0)},L_{1}^{(1)},\hat{u}_{0})=(-1)^{\hat{u}_{0}}L_{1}^{(0)} + L_{1}^{(1)}.%\label{eq3}
	\end{flalign*}
	
	Finally, %$\hat{u}_{1}=h(L(u_1))$. %
	a hard decision on $ L(u_1) $ produces the estimate $ \hat{u}_1 $.
	
	By recursively applying the \( f \)- and \( g \)-functions \( n \) times, SC decoding produces \(\hat{\mathbf{u}}_{0:N} = [\hat{u}_0, \cdots, \hat{u}_{N-1}]\) for a length-\(2^n \) polar code with $\hat{u}_{k}=u_{k}$ for all $k\in\mathscr{A}^{c}$.
%	By recursively applying the $f$- and $g$-functions $n$ times, SC decoding outputs $\hat{\mathbf{u}}_{0:K}=[\hat{u}_0, \cdots, \hat{u}_{K-1}]$ for the $K$ information bits of a length-$2^n$ polar code.
%	By recursively applying the $f$- and $g$-functions $n$ times, SC decoding produces the decoded sequence $\hat{\mathbf{u}}_{0:K}=[\hat{u}_0, \cdots, \hat{u}_{K-1}]$ for a length-$2^n$ polar code $\mathbf{u}_{0:N}$ with $K$ information bits.
%	By recursively applying the $f$- and $g$-functions $n$ times, SC decoding derives the decoding results for a length-$2^{n}$ polar code as $\hat{\mathbf{u}}_{0:K}=[\hat{u}_0,\cdots,\hat{u}_{K-1}]$. %with the frozen bits are all-zero.
	
	%SCL decoding is a further development of SC decoding. It generates a list of decoding paths with the SC decoder, and at every information bit position where a split occurs, only the top $L$ paths with the smallest path metrics are kept \cite{ref4}.
	SCL decoding extends SC decoding by maintaining a list of decoding paths. The $ L $ most reliable paths are kept after each path split at an information bit position \cite{ref4}.
	\subsection{Perturbation-enhanced SC/SCL decoding of Polar Codes}
	To facilitate future discussions, we describe the $y$-side perturbation-enhanced algorithm \cite{ref13}: when the CRC check fails, each bit of \( \mathbf{L}^{(0:N)}_{1}=\frac{2\mathbf{y}_{0:N}}{\sigma^{2}} \) is perturbed by adding an independent noise \( n_{p_{i}} \sim \mathscr{N}(0, \sigma_p^2) \). This process repeats until a codeword \( \hat{\mathbf{u}}_{0:N} \) passes the CRC check or \( T \) attempts are reached, at which point \( \hat{\mathbf{u}}_{0:N} \) is returned as the final result. This method greatly improves the performance of long polar codes \cite{ref13}, proving effective in practice.
	%\textcolor{red}{To facilitate future discussions, we provide a detailed description of the $y$-side perturbation-enhanced algorithm \cite{ref13}}: when the CRC check fails, each bit of \( \mathbf{L}^{(0:N)}_{1} = \frac{2\mathbf{y}_{0:N}}{\sigma^{2}} \) is perturbed by adding an independent perturbation \( n_{p_{i}} \sim \mathscr{N}(0, \sigma_p^2) \). This process continues until either the codeword \( \hat{\mathbf{u}}_{0:N} \) passes the CRC check or \( T \) attempts are reached, at which point \( \hat{\mathbf{u}}_{0:N} \) is returned as the final result. This method significantly improves the performance of long polar codes \cite{ref13}, making it highly effective for practical applications.
	\begin{algorithm}
	\caption{$y$-side perturbation-enhanced SC/SCL decoding} 
	\begin{algorithmic}[1] 
	\Require Original received LLRs $\mathbf{L}_{1}^{(0:N)}$, perturbation power $\sigma_{p}^{2}$, maximum attempts $T$, information set $\mathscr{A}$
	\Ensure The decoded codewords $\hat{\mathbf{u}}_{0:N}$		
	\State $\hat{\mathbf{u}}_{0:N}$ $\gets$ SCLDecoder($\mathbf{L}^{(0:N)}_{1}$,$\mathscr{A}$);		
	\If{CRCCheck$(\hat{\mathbf{u}}_{0:N})=\text{fail}$}
	\For{$t=1$ to $T$ }	    
	\State $\mathbf{L}^{(0:N)}_{p} \gets \mathbf{L}^{(0:N)}_{1}+\mathbf{n}_{p}$ with $\mathbf{n}_{p}=[n_{p_{0}},\cdots,n_{p_{N-1}}]$ and $\mathbf{L}^{(0:N)}_{p}=[\frac{2y_{0}}{\sigma^{2}}+n_{p_{0}},\cdots,\frac{2y_{N-1}}{\sigma^{2}}+n_{p_{N-1}}]$;
	\State $\hat{\mathbf{u}}_{0:N}$ $\gets$ SCLDecoder($\mathbf{L}^{(0:N)}_{p}$,$\mathscr{A}$);	\If{CRCCheck$(\hat{\mathbf{u}}_{0:N})=\text{success}$}	return $\hat{\mathbf{u}}_{0:N}$;
	%\State $\mathbf{break}$
	\EndIf
	\EndFor
	\EndIf
	\end{algorithmic}
	\end{algorithm}
	\section{Delay Probability of the First Error Position in Perturbation-Enhanced SC Decoders}%Examining the Delay Probability of the First Error Position in Perturbation-enhanced SC Decoders}
	%The focus of this section is on understanding perturbation/flip-enhanced SC decoders. We initiate our discussion by examining the equivalent $u$-side perturbation noise and then proceed to determine the delay probability of the first error position.
	%This section delves into the theory of perturbation-enhanced SC decoders. We start by assessing the equivalent $u$-side perturbation noise and then calculate the delay probability of the first error position.
	%$L_1^{(i)}$ and 
	This section begins with approximating the impact of perturbation on decision LLRs in part A, and in part B, we show that the perturbation either maintains or improves the SC decoder's performance, each with a probability of $\frac{1}{2}$ (Theorem \ref{thm1}).
	%In this section, we first approximate the perturbation noises on decision LLRs in part A and subsequently prove that adding perturbation either maintains or improves the SC decoder’s performance, each with a probability of $\frac{1}{2}$ (Theorem \ref{thm1}) in part B.
	%The aim of this section is to understand perturbation-enhanced SC decoder (Algorithm 1). We begin by approximating the perturbation noises on the decision LLRs and then calculate the delay probability of the first error position.
	
	Let \( L_N^{(i)} (i=0, \ldots, N-1) \) (or simply \( L_i \)) be the hard-decision side LLR of bit \( i \) in the SC algorithm (or the \( u \)-side LLR), computed by an SC decoder. The mean of \( L_N^{(i)} \) is denoted as \( \mu_N^{(i)} \) or \( \mu_i \) in short. The following assumption reasonably estimates the distribution of $\{L_{i}\big| \hat{\mathbf{u}}_{0:i}=\mathbf{u}_{0:i}\}$ \cite{ref14}.
	\begin{assumption}\label{asm1} (Gaussian Approximation (GA))  Assume the LLRs of each subchannel follow a Gaussian distribution with mean of half the variance \cite{ref14} and $\hat{\mathbf{u}}_{0:i}=\mathbf{u}_{0:i}$ (i.e., $\hat{u}_{0}=u_{0},\cdots,\hat{u}_{i-1}=u_{i-1}$) when decoding $u_i$, we obtain:
		\begin{flalign*}
			\{L_{i}\big| \hat{\mathbf{u}}_{0:i}=\mathbf{u}_{0:i}\}\sim\mathscr{N}(\mu_N^{(i)},2\mu_N^{(i)}),
		\end{flalign*}
		where $\mu_N^{(i)}$ (abbreviated as $\mu_{i}$) is recursively given as follows:
		\begin{flalign*}
			\mu^{(i)}_{N}=
				\begin{cases}
					\phi^{-1}(1-(1-\phi(\mu^{(\frac{i}{2})}_{\frac{N}{2}}))^{2}), &\text{if}\ i\in \{0,2,\cdots,N-2\}, \\
					2 \mu_{\frac{N}{2}}^{(\frac{i-1}{2})}, &\text{if}\ i\in\{1,3,\cdots,N-1\},
			\end{cases}
		\end{flalign*}
		with
		\begin{flalign*}
			\phi(x)\triangleq 1-\int_{-\infty}^{\infty}\frac{1}{\sqrt{4\pi x}}\tanh(\frac{t}{2})e^{-\frac{(t-x)^{2}}{4x}} \mathrm{d}t.
		\end{flalign*}
		
		Here, $\mu_{1}^{(i)}=\frac{2}{\sigma^{2}}$ for each $i$ within $\{0,1,\cdots,N-1\}$ \cite{ref14}.
		%and $\mu_{\frac{N}{2}}^{(l)}$ denotes the $l$-th LLR at the layer just before making the hard decision, satisfies $\mu_1^{(0)}=\frac{2}{\sigma^2}$ \cite{ref15}.
%		Here, $ i $ being even means it is at an even index in $\{0, \cdots, N-1\}$ and $\mu_{\frac{N}{2}}^{(k)}$ is the $k$-th LLR at the layer before hard decision.
		%Here, \( i \) being even indicates it is at an even index in $\{0, \cdots, N-1\}$.
	\end{assumption}
	%\begin{remark}\label{re1}
	%	If we assume that $u_{1:N}\neq \mathbf{0}_{N}$, Assumption 1 could be restated as follows:
	%	\begin{flalign*}
		%		\{L_{i}\big| \hat{u}_{1:i-1}=u_{1:i-1}\}\overset{\mathrm{d}}{=}\mathscr{N}((1-2u_{i})\mu_i,2\mu_i).
		%	\end{flalign*}
	%\end{remark}
	%Assuming LLRs of each subchannel follow a constrained Gaussian distribution where the mean is half the variance \cite{ref15}, and given $u_{1: i-1}=0$ when decoding bit $i$, t
	%\mu_N^{(i)}$ can be recursively determined as follows:
	%\begin{flalign}
	%	\mu^{(i)}_{N} & =\phi^{-1}(1-(1-\phi(\mu^{(\frac{i+1}{2})}_{\frac{N}{2}}))^{2})
	%\end{flalign}
	%if $i$ is odd, where
	%\begin{flalign}
	%	\phi(x)=1-\int_{-\infty}^{\infty}\frac{1}{\sqrt{4\pi x}}\tanh(t/2)e^{-\frac{(t-x)^{2}}{4x}} \mathrm{d}t
	%\end{flalign}
	%
	%and
	%\begin{flalign}
	%	\mu^{(i)}_{N} & =2 \mu^{(\frac{i}{2})}_{\frac{N}{2}}
	%\end{flalign}
	%when $i$ is even.
	%Here are some key assumptions and properties for later use.
	\subsection{Approximating perturbations on the decision LLRs}
	%Based on Proposition \ref{prop1}, we propose an enhanced SC algorithm using $u$-side LLR perturbations $n_{u}$. This method directly perturbs decision LLRs in the SC algorithm, improving performance by changing the reliability assessments of bits on the hard-decision side. %altering bit reliability during hard decisions.
	In this section, we present a $u$-side perturbation method (Algorithm 2) that applies approximate perturbations to the decision LLRs $\mathbf{L}_{0:N}=[L_{0},\cdots,L_{N-1}]$ as detailed below.
	\begin{proposition}\label{prop1}
		Perturbations applied on the received LLRs $\mathbf{L}^{(0:N)}_{1}$ can be approximately viewed as adding perturbation noise $n_{u_{i}}$ to $L_{i}$ ($i=0, \cdots, N-1$), characterized by:
		%Perturbations applied on the $y$-side can approximately be seen as adding perturbation noise at bit $u_{i} (i=0,\cdots,N-1)$, characterized by:
		\begin{flalign}
			n_{u_{i}}\sim \mathscr{N}(0,\sigma^{2}_{i}),\label{eq5}
		\end{flalign}
		where $\sigma_{i}^{2}=2^{k_{i}}\sigma_{p}^{2}$ with $k_{i}$ the number of $g$ operations performed during SC decoding up to bit $u_{i}$. %Clearly, $\sigma_{i}^{2}\leq N\sigma_{p}^{2}$. %Then we denote $\mathbf{n}_{u}=\{n_{u_{i}}\}_{i=0}^{N-1}$ the $u$-side perturbation noises. %number of $f$ operations  of $i-1$. It is apparent that $\sigma_{1}^{2}=\sigma_{p}^{2}$.
	\end{proposition}
	\begin{proof}
		%Assume that $l_1$ and $l_2$ are independent and identically distributed (i.i.d.) from $\mathscr{N}(\mu, 2 \mu)$, 
		%Assuming polar code layers from left to right represent layer $0$ to layer $n$, consider a polarization unit in any layer. Let \( L^{+} \) and \( L^{-} \) be the LLR values of the upper and lower channels, respectively, with noises \( n_{p}^{+} \) and \( n_{p}^{-} \) applied. To estimate the effect of the perturbations applied on the received LLRs for SC decoding, assume \( L^{+} \) and \( L^{-} \) are independent, following \( N(\mu^{+}, 2\mu^{+}) \) and \( N(\mu^{-}, 2\mu^{-}) \). Noises \( n_{p}^{+} \) and \( n_{p}^{-} \) are independent, following \( N(0, \sigma_{p,+}^{2}) \) and \( N(0, \sigma_{p,-}^{2}) \), and are independent of \( L^{+} \) and \( L^{-} \).
		Consider a polar code with layers numbered from right (layer \(0\)) to left (layer \(n\)). In each layer \(k\), the input LLRs for the upper and lower polarized subchannels are \(l_k^{+}\) and \(l_k^{-}\). The perturbation noise are \(n_{k}^{+}\) for the upper and \(n_{k}^{-}\) for the lower, resulting in perturbed LLRs \(\tilde{l}_k^{+} = l_k^{+} + n_{k}^{+}\) and \(\tilde{l}_k^{-} = l_k^{-} + n_{k}^{-}\).%Consider a polar code with layers numbered from right (layer $0$) to left (layer \( n \)). In each layer \( k \), the input LLRs for the upper and lower polarized subchannels in a polarization unit are \( l_k^{+} \) and \( l_k^{-} \). Assuming that the perturbation noise affecting the polarization unit at layer \( k \) is given by \( n_{k}^{+} \) for the upper subchannel and \( n_{k}^{-} \) for the lower subchannel, the perturbed LLRs are then \( \tilde{l}_k^{+} = l_k^{+} + n_{k}^{+} \) and \( \tilde{l}_k^{-} = l_k^{-} + n_{k}^{-} \).}
			%Assume the layers of a polar code are numbered from right (layer 0) to left (layer \( n \)). In any layer \( k \), the input LLRs for the upper and lower polarized subchannels in a polarization unit are \( l_k^{+}\) and \( l_k^{-}\).	Assuming the perturbation noise applied to any polarization unit at layer \( k \) is \( n_{k}^{+} \) and \( n_{k}^{-} \), the perturbed input LLRs are \( \tilde{l}_k^{+} = l_k^{+} + n_{k}^{+} \) and \( \tilde{l}_k^{-} = l_k^{-} + n_{k}^{-} \).}
		
%		Assume the layers of a polar code are numbered from right (layer 0) to left (layer \(n\)). In layer \(k\), the input LLRs for the upper and lower channels are \(l_k^{+}\) and \(l_k^{-}\), respectively. With added noise terms \(n_{p,k}^{+}\) and \(n_{p,k}^{-}\), the perturbed LLRs are \(\tilde{l}_k^{+} = l_k^{+} + n_{p,k}^{+}\) and \(\tilde{l}_k^{-} = l_k^{-} + n_{p,k}^{-}\).
%		
%		According to the polarization rule, \(l_k^{+}\) and \(l_k^{-}\) are i.i.d. from \(\mathscr{N}(\mu_{(k)}, 2 \mu_{(k)})\), where \(\mu_{(k)}=\mathbb{E}(l_k^{+})\). The noise terms \(n_{p,k}^{+}\) and \(n_{p,k}^{-}\) are i.i.d. from \(\mathscr{N}(0, \sigma_{p,(k)}^{2})\), with \(\sigma_{p,(k)}^{2}\) as the variance of \(n_{p,k}^{+}\). Notably, \(n_{p,k}^{+}\) and \(n_{p,k}^{-}\) are independent of \(l_k^{+}\) and \(l_k^{-}\). For the initial layer (layer 0), \(n_{p,0}^{+}\) and \(n_{p,0}^{-}\) are i.i.d. from \(\mathscr{N}(0, \sigma_p^{2})\), and \(l_0^{+}\) and \(l_0^{-}\) are i.i.d. from \(\mathscr{N}(\frac{2}{\sigma^{2}}, \frac{4}{\sigma^{2}})\) \cite{ref14, ref15}.
		
		According to the polarization rule, \( l_k^{+} \) and \( l_k^{-} \) are independent and identically distributed (i.i.d.) from \(\mathscr{N}(\mu_{(k)}, 2 \mu_{(k)})\), where \(\mu_{(k)} = \mathbb{E}(l_k^{+})\). Similarly, \(n_{k}^{+}\) and \(n_{k}^{-}\) are i.i.d. from \(\mathscr{N}(0, \sigma_{(k)}^{2})\) with variance \(\sigma_{(k)}^{2}\), and are independent of \(l_k^{+}\) and \(l_k^{-}\). Specifically, \(n_{0}^{+}\) and \(n_{0}^{-}\) are i.i.d. from \(\mathscr{N}(0, \sigma_p^2)\), while \(l_0^{+}\) and \(l_0^{-}\) are i.i.d. from \(\mathscr{N}(\frac{2}{\sigma^2}, \frac{4}{\sigma^2})\) \cite{ref13, ref14}.%According to the polarization rule, \( l_k^{+} \) and \( l_k^{-} \) are independent and identically distributed (i.i.d.) from \(\mathscr{N}(\mu_{(k)}, 2 \mu_{(k)})\), with \(\mu_{(k)}=\mathbb{E}(l_k^{+})\). Similarly, \(n_{k}^{+}\) and \(n_{k}^{-}\) are i.i.d. from \(\mathscr{N}(0, \sigma_{(k)}^{2})\) with \(\sigma_{(k)}^{2}\) the variance, and are independent of \(l_k^{+}\) and \(l_k^{-}\). Specifically, \(n_{0}^{+}\) and \(n_{0}^{-}\) are i.i.d. from \(\mathscr{N}(0, \sigma_p^2)\), and \(l_0^{+}\), \(l_0^{-}\) are i.i.d. from \(\mathscr{N}(\frac{2}{\sigma^2}, \frac{4}{\sigma^2})\) \cite{ref13,ref14}.} %Notably, \( n_{p,k}^{+} \) and \( n_{p,k}^{-} \) are also independent of \( l_k^{+} \) and \( l_k^{-} \), while $n_{p,0}^{+}$ and $n_{p,0}^{-}$ are i.i.d from $\mathscr{N}(0,\sigma_{p}^{2})$ \cite{ref14}, \( l_0^{+} \) and \( l_0^{-} \) are i.i.d. from \(\mathscr{N}(\frac{2}{\sigma^{2}}, \frac{4}{\sigma^{2}})\) \cite{ref15}.}
		%By the polarization rule, \( l_k^{+} \) and \( l_k^{-} \) are independent and identically distributed (i.i.d.) from \(\mathscr{N}(\mu_{(k)}, 2 \mu_{(k)})\), with \(\mu_{(k)}\) the mean of \( l_k^{+} \). Similarly, \( n_{p,k}^{+} \) and \( n_{p,k}^{-} \) are i.i.d. from \(\mathscr{N}(0, \sigma_{p,k}^{2})\) with \(\sigma_{p,k}^{2}\) the variance of \( n_{p,k}^{+} \). Notably, \( n_{p,k}^{+} \) and \( n_{p,k}^{-} \) are also independent of \( l_k^{+} \) and \( l_k^{-} \). %\( l_0^{+} \) and \( l_0^{-} \), are the received LLRs i.i.d. from \(\mathscr{N}(\frac{2}{\sigma^{2}}, \frac{4}{\sigma^{2}})\).

%		\begin{figure}[H]
%			\centering
%			\includegraphics[width=2.55in]{main_idea.pdf}
%			\caption{The key idea for proving Proposition \ref{prop1} where $f(\cdot)$ and $g(\cdot)$ are defined in Section II B with $h(\cdot)$ the hard decision function.}
%			\label{fig0}
%		\end{figure}
		%The key for proving Proposition \ref{prop1} is to estimate the values of \( \square \) and \( \triangle \).
		To prove Proposition \ref{prop1}, it is essential to examine the relationships between $ \tilde{l}_{k+1}^{+} $ and $ l_{k+1}^{+} $, as well as $ \tilde{l}_{k+1}^{-} $ and $ l_{k+1}^{-}$.
		Using the multivariable Taylor expansion \cite{ref15}, we obtain:
		%Applying the multivariable Taylor expansion \cite{ref15}, we have: %can determine that:
		%To begin with, according to the multivariable Taylor expansion \cite{ref16}, we can deduce that
		\begin{flalign*}
			&\tilde{l}_{k+1}^{+}=f\left(\tilde{l}_{k}^{+}, \tilde{l}_{k}^{-}\right)=f\left(l_k^{+}+n^{+}_{k}, l_k^{-}+n^{-}_{k}\right)\\
			&\approx f\left(l_k^{+},l_k^{-}\right)+\frac{\partial f\left(l_k^{+}, l_k^{-}\right)}{\partial l_k^{+}} n^{+}_{k}+\frac{\partial f\left(l_k^{+}, l_k^{-}\right)}{\partial l_k^{-}} n^{-}_{k}\\
			&=l_{k+1}^{+}+\mathrm{sgn}(l_k^{-})\mathbb{I}_{|l_k^{+}|\leq|l_k^{-}|}n^{+}_{k}+\mathrm{sgn}(l_k^{+})\mathbb{I}_{|l_k^{-}|\leq|l_k^{+}|}n^{-}_{k}.%\label{eq7}
		\end{flalign*}
		
		Since \( n_{k}^{+} \) and \( n_{k}^{-} \) are i.i.d. from \(\mathscr{N}(0, \sigma_{(k)}^{2})\), we have:
		\begin{flalign*}
			&\mathrm{sgn}(l_k^{-})\mathbb{I}_{|l_k^{+}|\leq|l_k^{-}|}n^{+}_{k}+\mathrm{sgn}(l_k^{-})\mathbb{I}_{|l_k^{-}|\leq|l_k^{+}|}n^{-}_{k}\\
			&\sim\mathscr{N}(0,(\mathbb{I}_{|l_{k}^{-}|\leq|l_{k}^{+}|}+\mathbb{I}_{|l_{k}^{+}|\leq|l_{k}^{-}|})\sigma_{(k)}^{2})=\mathscr{N}(0,\sigma_{(k)}^{2}),
		\end{flalign*}
		where $\mathbb{P}(|l_{k}^{+}|=|l_{k}^{-}|)=0$ yields the last equality.
		
		Moreover, based on Assumption \ref{asm1} and the definitions of \( \tilde{l}_{k}^{+}, \tilde{l}_{k}^{-}, n^{+}_{k} \), and \( n^{-}_{k} \), we can state the following:
		\begin{flalign*}
			\tilde{l}_{k+1}^{-}=g(\tilde{l}_{k}^{+},\tilde{l}_{k}^{-},\hat{u}^{+})=l_{k+1}^{-}+(-1)^{\hat{u}^{+}}n^{+}_{k}+n^{-}_{k},%\label{eq6}
		\end{flalign*}
		where $\hat{u}^{+}=h(\tilde{l}_{k+1}^{+})$ and  $(-1)^{\hat{u}^{+}}n^{+}_{k}+n^{-}_{k}\sim\mathscr{N}(0,2\sigma_{(k)}^{2})$.
		
		%If we set $L^{+}$ and $L^{-}$ the received LLRs which means that $n_{p}^{+}$ and $n_{p}^{-}$ are i.i.d from $\mathscr{N}(0,\sigma_{p}^{2})$, then we could conclude that the perturbation power remains unchanged after an $f$ operation, but it doubles after a $g$ operation.
		%The above analysis yields that 
		
		%Figure  illustrates the transfer process of $\sigma_{p}^{2}$. 
		%The perturbation power remains unchanged after one $f$-operation, but doubles after one $g$-operation. By iterating this process, we can determine the equivalent perturbation distribution for $u_{i}$, as shown in (\ref{eq5}).
%		\begin{figure}[H]
%			\centering
%			\includegraphics[width=2.5in]{length_4_example.pdf}
%			\caption{An illustrative example of Proposition \ref{prop1} with $N=4$.}
%			\label{fig2}
%		\end{figure}
        The preceding analysis reveals that when moving perturbation noise from layer $ k $ to layer $ k+1 $, the variance stays the same after an $ f $ operation but doubles after a $ g $ operation.
        %The preceding analysis reveals that when moving perturbation noise from layer $ k $ to layer $ k+1 $, the noise power remains unchanged after an $ f $ operation but doubles after a $ g $ operation.

		Therefore, Proposition \ref{prop1} has been demonstrated.
		%expressed as:
		%\begin{equation}
		%	\sigma^{2,(k+1)}=
		%	\begin{cases}
			%		\sigma^{2,(k)},& \mathrm{if\ a\ } f\mathrm{-operation\ is\ performed},\\
			%		2\sigma^{2,(k)},& \mathrm{if\ a\ } g\mathrm{-operation\ is\ performed},
			%	\end{cases}
		%	\label{eq8}
		%\end{equation}
		%where $\sigma^{2,(k)}$ is the equivalent perturbation power at the $(n-k)$-th polarizing layer.
		%Hence, the proof of Proposition \ref{prop1} is established.
		%By iterating this process, we obtain the equivalent perturbation distributions at each layer. Specifically, the equivalent perturbation distribution of $u_{i}$ can be determined by (\ref{eq5}), i.e., $n_{u_{i}}\sim\mathscr{N}(0,\sigma^{2}_{i})$.
	\end{proof}%provides a new insight by suggesting
	Proposition \ref{prop1} supports Algorithm 2 ($u$-side perturbation-enhanced SC decoding), showing how perturbations affect decoding by altering decision LLRs $\mathbf{L}_{0:N}$, obtained through $n$ iterations of $f$ and $g$ operations on the received LLRs.
	%Proposition \ref{prop1} indicates that perturbations on received LLRs $\mathbf{L}^{(0:N)}_{1}$ can be approximated as those on decision LLRs $\mathbf{L}_{0:N}$. This suggests that perturbations impact performance by modifying the values of the decision LLRs. 
	%This informs us that perturbations affect decoding performance by altering the values of the decision LLRs. %values used for hard decisions.
	%This means that the perturbations affect performance by altering the decision LLRs.
	%Proposition \ref{prop1} offers a novel perspective by suggesting that perturbations on received LLRs can be approximately regarded as perturbations on decision LLRs $\mathbf{L}_{0:N}$. This insight reveals that the performance of the SC decoder is affected by changes in the decision LLRs during the SC decoding process. Consequently, it underscores the critical role perturbations play in influencing the overall decoding accuracy and reliability.
	%Proposition \ref{prop1} provides a new insight that perturbations on received LLRs can be approximately treated as those on decision LLRs $\mathbf{L}_{0:N}$, which informs us that perturbation affects the SC decoder's performanc by changing the decision LLRs in SC decoding.
	%enabling us to devise a $u$-side perturbation-enhanced SC decoding algorithm (Algorithm 2).
	\begin{algorithm}
		\caption{$u$-side perturbation-enhanced SC decoding} 
		\begin{algorithmic}[1]
			\Require Original received LLRs $\mathbf{L}_{1}^{(0:N)}=\frac{2\mathbf{y}_{0:N}}{\sigma^{2}}$, perturbation power $\sigma_{p}^{2}$, maximum attempts $T$, information set $\mathscr{A}$
			\Ensure The decoded codewords $\hat{\mathbf{u}}_{0:N}$
			\State $\mathbf{L}_{0:N}$ $\gets$ $\mathbf{L}_{1}^{(0:N)}$, $\{\sigma^{2}_{i}\}_{i=0}^{N-1} \gets \sigma_{p}^{2}$;
			\State $\hat{\mathbf{u}}_{0:N}$ $\gets$ $h(\mathbf{L}_{0:N})$ with $h(\cdot)$ the hard decision function;
			\If{CRCCheck$(\hat{\mathbf{u}}_{0:N})=\text{fail}$}
			\For{$t=1$ to $T$ }	    
			\State $\mathbf{L}'_{0:N} \gets \mathbf{L}_{0:N}+\mathbf{n}_{u}$ with $\mathbf{n}_{u}=[n_{u_{0}},\cdots,n_{u_{N-1}}]$ and $\mathbf{L}'_{0:N}=[L_{0}+n_{u_{0}},\cdots,L_{N-1}+n_{u_{N-1}}]$;  %, where $n_{u}=\{n_{u_0},\cdots,n_{u_{N-1}}\}$		
			\State $\hat{\mathbf{u}}_{0:N}$ $\gets$ $h(\mathbf{L}'_{0:N})$;		
			\If{CRCCheck$(\hat{\mathbf{u}}_{0:N})=\text{success}$} return $\hat{\mathbf{u}}_{0:N}$;	
			%\State $\mathbf{break}$
			\EndIf
			\EndFor
			\EndIf
		\end{algorithmic}
	\end{algorithm}
	\subsection{The delay probability of the first error position for $u$-side perturbation-enhanced SC decoders}
	In this section, using Lemma \ref{lem1} (upper bound on unchanged probability of the first error position) and Lemma \ref{lem2} (probability that the perturbed decoder is correct given the initial one is correct), we show that perturbation does not degrade SC decoder's performance and improves it with probability \(\frac{1}{2}\). %Lemma 1 (an upper bound on the unchanged probability of the first error position) and Lemma 2 (the probability that the perturbed decoder is correct when the initial decoder is correct) to prove that perturbation does not degrade SC decoder's performance and improves it with probability $\frac{1}{2}$.
	%In this section, we provide an upper bound on the probability that the first error position remains unchanged after a perturbation (Lemma 1). We then show that if the initial decoder is correct, the perturbed decoder is almost always correct (Lemma 2). Combining these results, we prove that perturbation does not degrade SC decoder performance and improves it with probability $\frac{1}{2}$.
	
	%In this section, our main focus is the asymptotic behavior of the following probabilities:
	%In this section, we assume that all-zero codewords $\mathbf{0}_{N}$ are transmitted throughout this paper for simplicity \cite{ref15} and we use $0_k$ to denote a length-$k$ all-zero vector. %, irrespective of its length.
	Assume that the all-zero codeword $\mathbf{0}_{N}$ is transmitted \cite{ref14}, and use $\mathbf{0}$ to represent an all-zero vector of any length. Let $\hat{\mathbf{u}}_{0:N}^{(0)}$ and $\hat{\mathbf{u}}_{0:N}^{(1)}$ be the decoded results of the SC decoder and the $u$-side perturbation-enhanced SC decoder (Algorithm 2 with $T=1$) as described in Section II-B, respectively. Define $\tau(0)$ and $\tau(1)$ as their first error positions. Then, we have: %Define $\tau(0)$ and $\tau(1)$ as the first error positions of these decoders. Then, we have:
	%Denote $\hat{\mathbf{u}}_{0:N}^{(0)}$ and $\hat{\mathbf{u}}_{0:N}^{(1)}$ the decoding results of the SC decoder and the perturbation-enhanced SC decoder (Algorithm 2 with $T=1$) as in Section II-B. Let $\tau(0)$ and $\tau(1)$ define the first error positions of these decoders. We assume that all-zero codewords \( \mathbf{0}_{N} \) are transmitted \cite{ref15} and use $\mathbf{0}$ to represent an all-zero vector, regardless of its length. Then, we have:
	\begin{flalign*}
		&\{\tau(0)=i\}=\{\hat{u}^{(0)}_{i}=1,\hat{\mathbf{u}}^{(0)}_{0:i}=\mathbf{0}\},\\
		&\{\tau(1)=i\}=\{\hat{u}^{(1)}_{i}=1,\hat{\mathbf{u}}^{(1)}_{0:i}=\mathbf{0}\}.
	\end{flalign*}
	
	This section examines the asymptotic behavior of delay probability (\ref{eqi}), unchanged probability (\ref{eqii}), and advance probability (\ref{eqiii}), given an SC decoding failure \(\{\tau(0)<N\}\). %By combining Lemma 1 (an upper bound on the unchanged probability) and Lemma 2 (the probability that the perturbed decoder is correct given the initial one is correct), we prove that perturbation does not degrade the performance of the SC decoder and improves it with a probability of \(\frac{1}{2}\).
	%This section focuses on analyzing the asymptotic behavior of the delay probability (\ref{eqi}) and the advance probability (\ref{eqii}). %and the advance probability.
	%This section is devoted to examining the asymptotic behavior of the following delay probability and unchanged probability:
	\begin{flalign}
		&\mathbb{P}(\tau(1)>\tau(0)\big|\tau(0)<N),\label{eqi}\\
		&\mathbb{P}(\tau(1)=\tau(0)\big|\tau(0)<N),\label{eqii}\\
		&\mathbb{P}(\tau(1)<\tau(0)\big|\tau(0)<N).\label{eqiii}
	\end{flalign}
	
	%	Since $\tau(0)<N$ indicates an SC decoding failure, (\ref{eqi}) and (\ref{eqii}) describe the conditional probabilities of the first error position is either delayed or unchanged, given this failure. By observing that early errors degrade SC decoder's performance, %the first error occurring sooner means decoding errors appear at earlier bits, which subsequently degrades the performance of the SC decoder, 
	%	we intend to introduce perturbations that minimize this probability and aim for avoiding errors for the SC decoder which can only be accomplished by delaying the first error position. %and ensure the first error does not occur before $N$ (this is achievable only by delaying the first error position!). %These probabilities help assess whether perturbations enhance SC decoder performance and quantify the improvements. 
	%determine if perturbations improve SC decoder performance and quantify the improvements.	%Note that the event $\tau(0)<\infty$ implies the failure of the original SC decoding. Therefore, (\ref{eqi}) and (\ref{eqii}) represent the conditional probabilities that the first error position is delayed and remains unchanged, respectively, given that the original SC decoding has failed.
	%which we term the delay probability and the unchanged probability.%The SC decoder works correctly only if $\tau(0)\geq N$. 
	
	%Notice that (\ref{eqi}) and (\ref{eqii}) describe the conditional probabilities of the first error position being either delayed or advanced, \(\{\tau(0)<N\}\) indicates an SC decoding failure, then  given this failure. 
	To avoid performance degradation, it is crucial that (\ref{eqiii}) diminishes as \(N \rightarrow +\infty\), because an earlier first error position implies incorrect decoding of earlier bits.

	Let's start with some key findings. Assuming \(\sigma_{p}^{2} \approx N^{-\alpha}\) with \(\alpha \in (1-2\beta, \beta)\) for some \(\beta \in \left(\frac{1}{3}, \frac{1}{2}\right)\), we focus on an information set \(\mathscr{A} \subset \left\{ i \big| Z\left(W_N^{(i)}\right) < 2^{-N^\beta} \right\}\) for clarity, This implies $\mu_i \geq N^{\beta}, \forall i \in \mathscr{A}$, where \(Z(W)\) is the Bhattacharyya parameter of \(W\) and \(W_{N}^{(i)}\) is the \(i\)-th bit-channel \cite{ref1}.
	%To begin, we present some useful results. Assume \(\sigma_{p}^{2} \approx N^{-\alpha}\) with \(\alpha \in (1-2\beta, \beta)\) for some \(\beta \in \left(\frac{1}{3}, \frac{1}{2}\right)\). For clarity, we consider an information set \(\mathscr{A} \subset \left\{ i \mid Z\left(W_N^{(i)}\right) < 2^{-N^\beta} \right\}\), implying \(\mu_i \geq N^{\beta}\) for all \(i \in \mathscr{A}\), where \(Z(W)\) is the Bhattacharyya parameter of \(W\) with \(W_{N}^{(i)}\) the \(i\)-th bit-channel \cite{ref1}. %and \(W_{N}^{(i)}\) denotes the \(i\)-th split bit-channel \cite{ref1}.
	%To start, we present some helpful results. Assume $\sigma_{p}^{2}\approx N^{-\alpha}$ with $\alpha\in(1-2\beta, \beta)$ for some $\beta\in(\frac{1}{3}, \frac{1}{2})$. For clarity, we consider an information set $\mathscr{A} \subset\left\{i \big| Z\left(W_N^{(i)}\right)<2^{-N^\beta}\right\}$, which implies $\mu_i \geq N^{\beta}$ for all $i \in \mathscr{A}$, where $Z(W)$ is the Bhattacharyya parameter of $W$ and $W_{N}^{(i)}$ the $i$-th splitted bit-channel \cite{ref1}.
	%We start by presenting some useful lemmas. Let $\sigma_{p}^{2}\approx N^{-\alpha}$ with $\alpha\in(1-2 \beta, \beta)$ for some $\beta\in(\frac{1}{3},\frac{1}{2})$. To provide clarity in the analysis, we consider an information set $\mathscr{A} \subset\left\{i \big| Z\left(W_N^{(i)}\right)<2^{-N^\beta}\right\}$ which tells us that $\mu_i \geq N^{\beta}$ for $\forall i\in\mathscr{A}$, where $Z(W)$ is the Bhattacharyya parameter of $W$.
	%, and based on this assumption, the following properties can be derived:
	
	The proposition below offers a lower bound of $\sigma_{i}^{2}$.
	%The following remark gives us an upper bound of the perturbation power described in Poposition \ref{prop1}.
	\begin{proposition}\label{prop2}
		For any $\gamma<\beta$ and $\sigma_{p}^{2}\approx N^{-\alpha}$, we find:
		\begin{flalign}
			\sigma_{i}^{2}\geq N^{\gamma-\alpha}, \forall i \in \mathscr{A}.
		\end{flalign}
	\end{proposition}
	\begin{proof}
		Let $\tilde{\mathscr{A}}=\left\{i\big| k_i \geq n\gamma\right\}$. Considering the recursion of $Z_n$ where $Z_{n}=Z_{n-1}^2$ if $B_{n}=1$ and $Z_{n} \geq Z_{n-1}$ if $B_{n}=0$, with $B_{n}=1$ denoting a $g$-operation in the $n$-th layer, we get:	
		\begin{flalign*}
			Z\left(W_N^{(i)}\right) \geq Z_{0}^{2^{n\gamma}}, \quad \forall i \in \mathscr{A}^c,
		\end{flalign*}
		where \(Z_{0}\) is the Bhattacharyya parameter of a BI-AWGN channel with noise \(\mathscr{N}(0, \sigma^{2})\). Given that $\gamma<\beta$, then, for sufficiently large $N$, $Z_0^{2^{n\gamma}}>2^{-N^\beta}$, which implies $\mathscr{A} \subset \tilde{\mathscr{A}}$.
		
		Consequently, $\sigma_i^2=2^{k_{i}}\sigma_{p}^{2} \geq 2^{n\gamma} \sigma_p^2\approx N^{\gamma-\alpha}$.
	\end{proof}
	%The next assumption, based on Proposition \ref{prop1}, offers a generalized Gaussian approximation for the SC+perturbation decoder which will be used to identify (\ref{eqi}) and (\ref{eqii}).
	Proposition \ref{prop1} allows us to derive the next property, termed the generalized GA for the $u$-side perturbation-enhanced SC decoder, which will be used to calculate (\ref{eqi})-(\ref{eqiii}).
	%The next property, termed as the generalized GA for the perturbation-enhanced SC decoder (Algorithm 1), can be derived from Proposition \ref{prop1}. This property will be utilized to determine the probabilities presented in (\ref{eqi})-(\ref{eqiii}).
	\begin{property}\label{pro1} For $\forall i\in\mathscr{A}$, we have:
		\begin{flalign*}
			\{L_{i},L_{i}+n_{u_{i}}\big| \hat{\mathbf{u}}^{(0)}_{0:i}=\mathbf{0},\hat{\mathbf{u}}^{(1)}_{0:i}=\mathbf{0} \}\sim (G_{1},G_{1}+G_{2}),
		\end{flalign*}
		where $G_{1}\sim\mathscr{N}(\mu_i,2\mu_i)$, $G_{2}\sim\mathscr{N}(0,\sigma_{i}^{2})$ independent of $G_{1}$.
	\end{property}
	%From Remark \ref{re1}, we obtain the following property. %Comprehensive details are available in \cite{ref16} due to space limitations.
	%Consider an information set $\mathscr{A} \subset\left\{i \big| Z\left(W_N^{(i)}\right)<2^{-N^\beta}\right\}$, we know that $\mu_{i}\geq N^{\beta}$. Proposition \ref{prop1} tells us that $\sigma_{i}^{2}=2^{k_{i}}\sigma_{p}^{2}\geq N^{\gamma-\alpha}$ for $\forall \gamma<\beta$. Then, we have:
	
	%We present some useful lemmas first and we assume that $\sigma_{p}^{2}\approx N^{-\alpha}$ with $\alpha\in (1-2\beta,\beta)$ for some $\beta\in (1/3,1/2)$.  We consider a given information set $\mathscr{A}\subset \{i\big| Z(W_{N}^{(i)})< 2^{-N^{\beta}}\}$ where $Z(W)$ is the  Bhattacharyya parameter of $W$. %and a given information set $\mathscr{A}\subset \{i\big| Z(W_{N}^{(i)})< 2^{-N^{\beta}}\}$ where $Z(W)$ is the  Bhattacharyya parameter of $W$.
	The proposition described below illustrates the probability that errors cannot be corrected after one perturbation.
	\begin{proposition}\label{prop3}
		Let $\mu,\ \sigma>0$ and $L \sim \mathscr{N}(\mu, 2 \mu), n_L \sim \mathscr{N}\left(0, \sigma_L^2\right)$, with $\mu=\mathbb{E}[L]$. If $L$ and $n_{L}$ are independent, then:
		\begin{flalign*}
			\mathbb{P}(L+n_{L}<0\big| L<0)\leq \inf _{0<s<\frac{1}{2}}\left[\frac{1}{2}+s+\frac{2+\mu}{\sqrt{2 \mu}} e^{-\sqrt{\frac{\pi}{2}} \sigma_{L} s}\right] .
		\end{flalign*}
	\end{proposition}
	\begin{proof}
		For any $s\in(0,\frac{1}{2})$, let $t=\frac{1}{2}-s\in (0,\frac{1}{2})$, we have:
		\begin{flalign*}
			&\mathbb{P}(L<0, L+n_{L}<0)=\mathbb{E}\left[\mathrm{Q}\left(\frac{L}{\sigma_{L}}\right) \mathbb{I}_{\{L<0\}}\right] \\
			=&\mathbb{E}\left[\mathrm{Q}\left(\frac{L}{\sigma_{L}}\right) \mathbb{I}_{\{-\mathrm{Q}^{-1}(t)<\frac{L}{\sigma_{L}}<0\}}+\mathrm{Q}\left(\frac{L}{\sigma_{L}}\right) \mathbb{I}_{\{\frac{L}{\sigma_{L}}<-\mathrm{Q}^{-1}(t)\}}\right] \\
			\leq& Q\left(-\mathrm{Q}^{-1}(t)\right) \mathbb{P}\left(-\sigma_{L} \mathrm{Q}^{-1}(t)<L<0\right)\\
			+&\mathbb{P}\left(L<-\sigma_{L} \sqrt{2 \pi}\left(\frac{1}{2}-t\right)\right) \\
			\leq&(1-t) \mathbb{P}(L<0)+\mathrm{Q}\left(\frac{\sigma_{L} \sqrt{2 \pi}\left(\frac{1}{2}-t\right)+\mu}{\sqrt{2 \mu}}\right) \\
			=&\left(\frac{1}{2}+s\right) \mathbb{P}(L<0)+\mathrm{Q}\left(\frac{\sqrt{2 \pi} \sigma_{L} s+\mu}{\sqrt{2 \mu}}\right),
		\end{flalign*}
		where $\mathrm{Q}^{-1}(x)\geq \sqrt{2\pi}(\frac{1}{2}-x)$ \cite{ref15} tells us the first inequality.
		
		Since $\frac{x}{\sqrt{2\pi}(1 + x^2)} e^{-\frac{x^2}{2}} \leq Q(x) \leq \frac{1}{\sqrt{2\pi} x} e^{-\frac{x^2}{2}}$ \cite{ref15}, we have:
		\begin{flalign*}
			\frac{\mathrm{Q}\left(\frac{\sqrt{2 \pi} \sigma_{L} s+\mu}{\sqrt{2 \mu}}\right)}{\mathbb{P}(L<0)}=\frac{\mathrm{Q}\left(\frac{\sqrt{2 \pi} \sigma_{L} s+\mu}{\sqrt{2 \mu}}\right)}{\mathrm{Q}(\sqrt{\frac{\mu}{2}})}\leq \frac{2+\mu}{\sqrt{2 \mu}} e^{-\sqrt{\frac{\pi}{2}} \sigma_{L} s},
		\end{flalign*}
		from which Proposition \ref{prop3} follows.
	\end{proof}
	From Proposition \ref{prop3}, we can further derive the upper bound of the probability that the first error position remains unchanged after one perturbation.
	%The following lemma describes the probability that the first error position remains unchanged after a perturbation.
	\begin{lemma}\label{lem1}
		Let $i \in \mathscr{A}$, then for any $\gamma<\beta, \alpha\in(1-2\beta,\beta)$ and sufficiently large $N$, we obtain:
		\begin{flalign*}
			\mathbb{P}(\tau(1)=i\big| \tau(0)=i)
			\leq \inf _{0<s<\frac{1}{2}}\left[\frac{1}{2}+s+\frac{2 N}{\sigma^2}e^{-\sqrt{\frac{\pi}{2}}  N^{\frac{\gamma-\alpha}{2}}s} \right]. 
		\end{flalign*}
	\end{lemma}
	\begin{proof}
		As stated in Proposition \ref{prop3}, it follows that:
		\begin{flalign*}
			\begin{aligned}
				&\mathbb{P}(\tau(1)=i, \tau(0)=i) \\
				=&\mathbb{P}\left(\hat{u}_{i}^{(1)}=1, \hat{\mathbf{u}}_{0: i}^{(1)}=\mathbf{0}, \hat{u}_{i}^{(0)}=1, \hat{\mathbf{u}}_{0: i}^{(0)}=\mathbf{0}\right) \\
				=&\mathbb{P}\left(\hat{u}_{i}^{(0)}=1, \hat{u}_{i}^{(1)}=1 \big| \hat{\mathbf{u}}_{0: i}^{(0)}=\mathbf{0}, \hat{\mathbf{u}}_{0: i}^{(1)}=\mathbf{0}\right)\\ 
				&\mathbb{P}\left(\hat{\mathbf{u}}_{0: i}^{(0)}=\mathbf{0}, \hat{\mathbf{u}}_{0: i}^{(1)}=\mathbf{0}\right)\\
				\leq& \mathbb{P}\left(L_i<0, L_i+n_{u_{i}}<0\right) \mathbb{P}\left(\hat{\mathbf{u}}_{0: i}^{(0)}=\mathbf{0}\right)\leq \mathbb{P}\left(L_i<0\right)\\
				&\inf _{0<s<\frac{1}{2}}\left[\frac{1}{2}+s+\frac{2+\mu_i}{\sqrt{2 \mu_i}} e^{-\sqrt{\frac{\pi}{2}}\sigma_{i} s}\right]\mathbb{P}\left(\hat{\mathbf{u}}_{0: i}^{(0)}=\mathbf{0}\right).
			\end{aligned}
		\end{flalign*}
		%where the last inequality holds according to Property \ref{pro1} and Proposition \ref{prop3}.
		
		Alternatively, Proposition \ref{prop2} infers that $\sigma_i^2 \geq N^{\gamma-\alpha}$. Also, for $\forall i \in \mathscr{A}, \sqrt{2 \mu_i} \geq 1$ and $\mu_i \leq 2 N / \sigma^2$. Notably, $\mathbb{P}(\tau(0)=i)=\mathbb{P}\left(L_i<0\right) \mathbb{P}\left(\hat{\mathbf{u}}_{0: i}^{(0)}=\mathbf{0}\right)$, and $2e^{-\sqrt{\frac{\pi}{2}}  N^{\frac{\gamma-\alpha}{2}}s}$ is small for sufficiently large $N$, thereby concluding the proof.
		%From another point of view, Remark \ref{re1} implies that $\sigma_{i}^{2}\geq N^{\gamma-\alpha}$. Moreover, $\sqrt{2 \mu_i} \geq 1$ for $i \in \mathscr{A}$ and $\mu_i \leq 2 N / \sigma^2$ for all $i$. Note that $\mathbb{P}(\tau(0)=i)=\mathbb{P}\left(L_i<0\right)\mathbb{P}\left(\hat{u}_{1: i-1}^{(0)}=0\right)$, thereby completing the proof.
	\end{proof}

	To prove Lemma 2, we need the distribution of \(\{L_i \big| \hat{\mathbf{u}}_{0: i}^{(0)} = \mathbf{0}, \hat{\mathbf{u}}_{0: j}^{(1)} = \mathbf{0}\}\) for \(i, j \in \mathscr{A}\) with \(j < i\). Due to the complexity of calculating this distribution, we approximate it using the distribution of \(\{L_i \big| \hat{\mathbf{u}}_{0: i}^{(0)} = \mathbf{0}\}\). For large $N$, we assume $\sigma_{p}^{2}\approx N^{-\alpha}$. Using the law of total probability, we have: $\mathbb{P}(L_i \in A \big| \hat{\mathbf{u}}_{0: i}^{(0)} = \mathbf{0})=\mathbb{P}(L_i \in A,\hat{\mathbf{u}}^{(1)}_{0:j}\neq \mathbf{0}  \big| \hat{\mathbf{u}}_{0: i}^{(0)} = \mathbf{0})+\mathbb{P}(L_i\in A,\hat{\mathbf{u}}^{(1)}_{0:j}=\mathbf{0}  \big| \hat{\mathbf{u}}_{0: i}^{(0)} = \mathbf{0})$ where $\mathbf{u}^{(1)}_{0:j}\neq \mathbf{0}$ indicates that there is a $k\in\mathscr{A}$ and $k<j$, s.t., $\hat{u}^{(1)}_{k}=1$. Then, we can show that for any $ j < i$, and any Borel set \(A\), $|\mathbb{P}(L_i \in A \big| \hat{\mathbf{u}}_{0: i}^{(0)} = \mathbf{0}, \hat{\mathbf{u}}_{0: j}^{(1)} = \mathbf{0}) - \mathbb{P}(L_i \in A \big| \hat{\mathbf{u}}_{0: i}^{(0)} = \mathbf{0})| < 2^{-N^{\beta}}$ for some $\beta\in (\frac{1}{3},\frac{1}{2})$ for sufficiently large $N$. This leads to the following assumption. %From Proposition \ref{prop1}, we obtain \(\sigma_i^2 \leq N^{1-\alpha}\), and the selection of \(\mathscr{A}\) ensures \(\mu_i \geq N^{\beta}\). As a result, for the SC decoder, the perturbation's effect on the decision LLRs of correctly decoded bits is small, as the perturbed values are dominated by $N^{\beta}$. Thus, the perturbation does not substantially alter the decoder's output, leading us to conclude the following result.} %for any \(j < i\), any Borel set \(A\), and any \(\varepsilon > 0\), we can show that \(| \mathbb{P}(L_i \in A \big| \hat{u}_{0: i}^{(0)} = \mathbf{0}, \hat{u}_{0: j}^{(1)} = \mathbf{0}) - \mathbb{P}(L_i \in A \big| \hat{u}_{0: i}^{(0)} = \mathbf{0})| < \varepsilon\) for large \(N\). Thus, we get the following Proposition.} %approximate this distribution using that of \(\{L_i \big| \hat{u}_{0: i}^{(0)} = \mathbf{0}\}\).}
	%we can show that for any $ j < i$, any Borel set \(A\), and any \(\varepsilon > 0\), $\big|\mathbb{P}(L_i \in A \big| \hat{u}_{0: i}^{(0)} = \mathbf{0}, \hat{u}_{0: j}^{(1)} = \mathbf{0}) - \mathbb{P}(L_i \in A \big| \hat{u}_{0: i}^{(0)} = \mathbf{0})\big| < \varepsilon$ for sufficiently large $N$ (see \cite{ref17} for more details due to space limitations). Thus, we can approximate it using the distribution of \(\{L_i \big| \hat{u}_{0: i}^{(0)} = \mathbf{0}\}\).}
	\begin{assumption}\label{asum2}
	   For any $ i, j \in \mathscr{A} $ with $ j < i $, for any Borel set $ A $ and sufficiently large $ N $, with $\sigma_{p}^{2}\approx N^{-\alpha}$, we can estimate:
	%For $\forall i,j \in \mathscr{A}$ with $j<i$, we estimate for any Borel set $A$ and sufficiently large $N$ that:
	   \begin{flalign*}
		   \mathbb{P}(L_i\in A \big| \hat{\mathbf{u}}_{0: i}^{(0)}=\mathbf{0}, \hat{\mathbf{u}}_{0: j}^{(1)}=\mathbf{0})\approx \mathbb{P}(L_i\in A \big| \hat{\mathbf{u}}_{0: i}^{(0)}=\mathbf{0}).
	   \end{flalign*}
    \end{assumption}
	Lemma 2 indicates that the perturbation decoder is unlikely to make an error if the first decoder didn't make an error. %if the original first error position is at bit $i$, it is highly probable that the first $j (j<i)$ bits will be correctly decoded after the perturbation.
	%The lemma that follows reveals that if the original error position is at bit $i$, then there is a high probability that the first $j (j<i)$ bits will still be decoded correctly after perturbation.
	%The subsequent lemma indicates that if the original first error position is bit $i$, then even after perturbation, there remains a high probability that the first $i-1$ bits will still be decoded correctly.
	\begin{lemma}\label{lem2}
		For any $i, j \in \mathscr{A}$ and $j<i$, the following applies:
		\begin{flalign*}
			\lim_{N\rightarrow \infty}\mathbb{P}(\hat{u}_{j}^{(1)}=0 \big| \tau(0)=i,\hat{\mathbf{u}}^{(1)}_{0:j}=\mathbf{0}) = 1.
			%\geq 1-3N e^{-N^{\alpha+2 \beta-1}},
		\end{flalign*}
		%when $N$ is sufficiently large.
	\end{lemma}
	\begin{proof}
		Employing the Bayesian formula and the conditional probability formula, we achieve for all $ j \in \mathscr{A} $ with $ j < i $:
		%By using Bayesian fomula and conditional probability fomula, we obtain for $\forall j<i, j\in\mathscr{A}$
		\begin{flalign*}
			\begin{aligned}
				&\mathbb{P}(\hat{u}_j^{(1)}=0 \big| \tau(0)=i,\hat{\mathbf{u}}^{(1)}_{0:j}=\mathbf{0})
				=\frac{\mathbb{P}(\tau(0)=i,\hat{\mathbf{u}}^{(1)}_{0:j+1}=\mathbf{0})}{\mathbb{P}(\tau(0)=i,\hat{\mathbf{u}}^{(1)}_{0:j}=\mathbf{0})} \\
				=& \frac{\mathbb{P}(\hat{\mathbf{u}}^{(0)}_{0:j}=\mathbf{0},\hat{\mathbf{u}}^{(1)}_{0:j+1}=\mathbf{0}\big| \hat{u}^{(0)}_{i} =1, \hat{\mathbf{u}}^{(0)}_{j:i} =\mathbf{0} )}{\mathbb{P}(\hat{\mathbf{u}}^{(0)}_{0:j}=\mathbf{0},\hat{\mathbf{u}}^{(1)}_{0:j}=\mathbf{0}\big| \hat{u}^{(0)}_{i} =1, \hat{\mathbf{u}}^{(0)}_{j:i} =\mathbf{0} )}\\
				=& \frac{\mathbb{P}(\hat{u}^{(0)}_i =1\big|\hat{\mathbf{u}}^{(0)}_{0:i}=\mathbf{0},\hat{\mathbf{u}}^{(1)}_{0:j+1}=\mathbf{0})\mathbb{P}(\hat{\mathbf{u}}^{(0)}_{0:j}=\mathbf{0},\hat{\mathbf{u}}^{(1)}_{0:j+1}=\mathbf{0})}{\mathbb{P}(\hat{u}^{(0)}_i=1\big|\hat{\mathbf{u}}^{(0)}_{0:i}=\mathbf{0},\hat{\mathbf{u}}^{(1)}_{0:j}=\mathbf{0})\mathbb{P}(\hat{\mathbf{u}}^{(0)}_{0:j}=\mathbf{0},\hat{\mathbf{u}}^{(1)}_{0:j}=\mathbf{0})}\\
				\overset{(a)}{\approx}&\frac{\mathbb{P}(\hat{\mathbf{u}}^{(0)}_{0:j}=\mathbf{0},\hat{\mathbf{u}}^{(1)}_{0:j+1}=\mathbf{0})}{\mathbb{P}(\hat{\mathbf{u}}^{(0)}_{0:j}=\mathbf{0},\hat{\mathbf{u}}^{(1)}_{0:j}=\mathbf{0})}\geq \mathbb{P}(\hat{\mathbf{u}}^{(0)}_{0:N}=\mathbf{0}_{N},\hat{\mathbf{u}}^{(1)}_{0:N}=\mathbf{0}_{N})\\
				%\times& \frac{\mathbb{P}\left(\hat{\mathbf{u}}^{(0)}_{0:j}=\mathbf{0},\hat{\mathbf{u}}^{(1)}_{0:j+1}=\mathbf{0}\right)}{\mathbb{P}\left(\hat{\mathbf{u}}^{(0)}_{0:j}=\mathbf{0},\hat{\mathbf{u}}^{(1)}_{0:j}=\mathbf{0}\right)}
			\end{aligned}
		\end{flalign*}
		\begin{flalign*}
			\begin{aligned}
				\geq& 1-\mathbb{P}(\hat{\mathbf{u}}^{(0)}_{0:N} \neq \mathbf{0}_{N})-\mathbb{P}(\hat{\mathbf{u}}^{(1)}_{0:N} \neq \mathbf{0}_{N})\\
				\geq& \prod_{k \in \mathscr{A}} \mathbb{P}(L_k+n_{u_{k}}>0)-2^{-N^\beta},
			\end{aligned}
		\end{flalign*}
		where \(L_k \sim \mathscr{N}(\mu_k, 2\mu_k)\) and $(a)$ follows from Assumption \ref{asum2} while %\((a)\) follows from the fact that for \(\forall j < i, j \in \mathscr{A}\), $\mathbb{P}(\hat{u}^{(1)}_j = 1 \big| \hat{\mathbf{u}}^{(0)}_{0:i} = \mathbf{0}, \hat{\mathbf{u}}^{(1)}_{0:j} = \mathbf{0}) = \mathbb{P}(L_j + n_{u_j} < 0 \big| \hat{\mathbf{u}}^{(0)}_{0:i} = \mathbf{0}, \hat{\mathbf{u}}^{(1)}_{0:j} = \mathbf{0}) =  \mathrm{Q}\left(\frac{\mu_j}{\sqrt{\mu_j + \sigma_j^2}}\right)\overset{(i)}{\leq} \mathrm{Q}(N^{\frac{\alpha+2\beta-1}{2}})$ where $\alpha\in (1-2\beta,\beta)$. Here, \((i)\) is based on \(\sigma_i^2 = 2^{k_i} \sigma_p^2 \leq N\sigma_p^2 \approx N^{1-\alpha}\) for \(\forall i \in \mathscr{A}\). Additionally, 
		\(\mathbb{P}(\hat{\mathbf{u}}^{(0)}_{0:N} \neq \mathbf{0}_N) \leq 2^{-N^\beta}\) \cite{ref1} yields the last inequality. %Since we set $\lim _{N \rightarrow \infty} \sigma_{p}=0$, there almost no rate loss. 
		
		Note that for $\forall x>0$, $\mathrm{Q}(x)\leq e^{-\frac{x^{2}}{2}}< 2^{-\frac{x^{2}}{2}}$ \cite{ref15}, we have:
		\begin{flalign*}
			\mathbb{P}\left(L_k+n_{u_{k}}>0\right)=\mathrm{Q}(-\frac{\mu_k}{\sqrt{2 \mu_k+\sigma_k^2}})\geq 1-2^{-\frac{\mu_k^{2}}{2(2\mu_k+\sigma_k^{2})}}.
		\end{flalign*}
		
		Proposition \ref{prop1} implies \(\sigma_{k}^{2} \leq N^{1-\alpha}\) for large $N$, leading to $\frac{\mu_k^{2}}{2\mu_k+\sigma_k^{2}}\geq \frac{1}{2N^{-\beta}+N^{1-\alpha-2\beta}}\geq\frac{1}{3N^{1-\alpha-2\beta}}$, thus we have:
%		Proposition \ref{prop1} implies \(\sigma_{k}^{2} \leq N^{1-\alpha}\) for large \(N\), leading to:
		%From Proposition \ref{prop1}, we have $\sigma_{i}^{2}\leq N^{1-\alpha}$, then we get:
		\begin{flalign*}
			&\mathbb{P}\left(\hat{u}_j^{(1)}=0 \big| \tau(0)=i,\hat{\mathbf{u}}^{(1)}_{0:j}=\mathbf{0}\right)\geq \left(1-2^{-\frac{N^{\alpha+2 \beta-1}}{6}}\right)^N-\\
			&2^{-N^\beta}\overset{(b)}{\geq} 1-N2^{-\frac{N^{\alpha+2 \beta-1}}{6}}-2^{-N^\beta} \overset{(c)}{\geq} 1- 3N 2^{-\frac{N^{\alpha+2 \beta-1}}{6}},
		\end{flalign*}
		where (b) follows from the Bernoulli’s inequality \cite{ref15} and \(2^{\frac{N^{\alpha+2\beta-1}}{6}} \leq 2N 2^{N^{\beta}}\) yields (c), thus proving Lemma \ref{lem2}.
		%Therefore, Lemma \ref{lem2} is established.
	\end{proof}
	The following theorem indicates that, with $\sigma_{p}^{2}\approx N^{-\alpha}$, asymptotically, both the delay probability and the unchanged probability of the first error position converge to $\frac{1}{2}$.
	\begin{theorem}\label{thm1}
		For any sufficiently large $N$, we achieve:
		%For any sufficiently large $N$, we obtain:
		\begin{flalign}
			&\lim_{N\rightarrow \infty} \mathbb{P}(\tau(1)>\tau(0)\big| \tau(0)<N) = \frac{1}{2},\\
			&\lim_{N\rightarrow \infty} \mathbb{P}(\tau(1)=\tau(0)\big| \tau(0)<N) = \frac{1}{2}.
		\end{flalign}
	\end{theorem}
	\begin{proof}
		Note that \(\{L_i < 0, L_i + n_{u_i} > 0\} \subset \{L_i < 0, n_{u_i} > 0\}\). Based on Property \ref{pro1} and applying the chain rule of conditional probability, we can deduce the following:
		\begin{flalign*}
			\begin{aligned}
				& \mathbb{P}(\tau(1)>\tau(0) \big| \tau(0)<N)=\frac{\sum_{i \in \mathscr{A}} \mathbb{P}(\tau(1)>i, \tau(0)=i)}{\mathbb{P}(\tau(0)<N)} \\
				=& \frac{\sum_{i \in \mathscr{A}} \mathbb{P}\left(L_i<0, L_i+n_{u_{i}}>0\right) \mathbb{P}\left(\hat{\mathbf{u}}_{0: i}^{(0)}=\mathbf{0},\hat{\mathbf{u}}_{0: i}^{(1)}=\mathbf{0}\right)}{\mathbb{P}(\tau(0)<N)} \\
				\leq& \frac{\sum_{i \in \mathscr{A}} \mathbb{P}(L_i<0, n_{u_{i}}>0) \mathbb{P}(\hat{\mathbf{u}}_{0: i}^{(0)}=\mathbf{0})}{\mathbb{P}(\tau(0)<N)} \\
				=&\frac{\sum_{i \in \mathscr{A}} \mathbb{P}\left(n_{u_{i}}>0\right) \mathbb{P}\left(L_i<0\right) \mathbb{P}(\hat{\mathbf{u}}_{0: i}^{(0)}=\mathbf{0})}{\mathbb{P}(\tau(0)<N)} \\
				=&\frac{1}{2} \frac{\sum_{i \in \mathscr{A}} \mathbb{P}(\tau(0)=i)}{\mathbb{P}(\tau(0)<N)}=\frac{1}{2}.
			\end{aligned}
		\end{flalign*}
		
		Utilizing Lemma \ref{lem1} and Lemma \ref{lem2} with $\gamma=\frac{\alpha+\beta}{2}$, we find:
		\begin{flalign*}
			\begin{aligned}
				& \mathbb{P}(\tau(1)>\tau(0) \big| \tau(0)<N)\\
				=&\frac{\sum_{i \in \mathscr{A}} \mathbb{P}(\tau(1)>\tau(0) \big| \tau(0)=i) \mathbb{P}(\tau(0)=i)}{\mathbb{P}(\tau(0)<N)} \\
				\geq& \frac{\sum_{i \in \mathscr{A}} \mathbb{P}(\tau(0)=i)(1-\mathbb{P}(\tau(1)=i \big| \tau(0)=i)) }{\mathbb{P}(\tau(0)<N)}\\
				\times&\prod_{j<i, j \in \mathscr{A}} \mathbb{P}\left(\hat{u}_j^{(1)}=0 \big| \tau(0)=i,\hat{\mathbf{u}}^{(1)}_{0:j}=\mathbf{0}\right)\\
				\geq& \sup _{0<s<\frac{1}{2}}\left[\frac{1}{2}-s-\frac{2 N}{\sigma^2}e^{-\sqrt{\frac{\pi}{2}} N^{\frac{\beta-\alpha}{4}}  s}\right] (1-3 N^2 2^{-\frac{N^{\alpha+2 \beta-1}}{6}}).
			\end{aligned}
		\end{flalign*}
		%where the Bernoulli’s inequality \cite{ref16} implies (c).
		
		Letting $N\rightarrow +\infty$ and then allowing $s\rightarrow 0$, we determine:
		\begin{flalign*}
			\lim_{N\rightarrow \infty}\mathbb{P}(\tau(1)>\tau(0) \big| \tau(0)<N)= \frac{1}{2}.
		\end{flalign*}
		
		Moreover, Lemma \ref{lem2} yields that%Recalling Lemma 2, we obtain:
		\begin{flalign*}
			\lim_{N\rightarrow \infty}\mathbb{P}(\tau(1)<\tau(0)\big|\tau(0)<N) = 0.
		\end{flalign*}
		
		As a result, the proof of Theorem \ref{thm1} is complete.
	\end{proof}
	Theorem \ref{thm1} shows that perturbation ensures the first error bit will either stay in the same position or shift to a later one, each with a probability of $\frac{1}{2}$. This implies that perturbation either maintains or improves the SC decoder's performance.
	\section{Simulation and discussions}
\begin{figure*}[htbp]
	\centering
	\subfigure[BLER performance comparisons between different perturbation-enhanced SC decoders: \( y \)-side (Algorithm 1) and \( u \)-side (Algorithm 2) for a code \((N, K+crc)=(1024, 512+24)\) with maximum attempts \( T = \{1, 5, 10\} \) and CRC generator polynomial \( g(x) = x^{24} + x^{23} + x^{6} + x^{5} + x + 1 \).]{
		\includegraphics[width=2.7in]{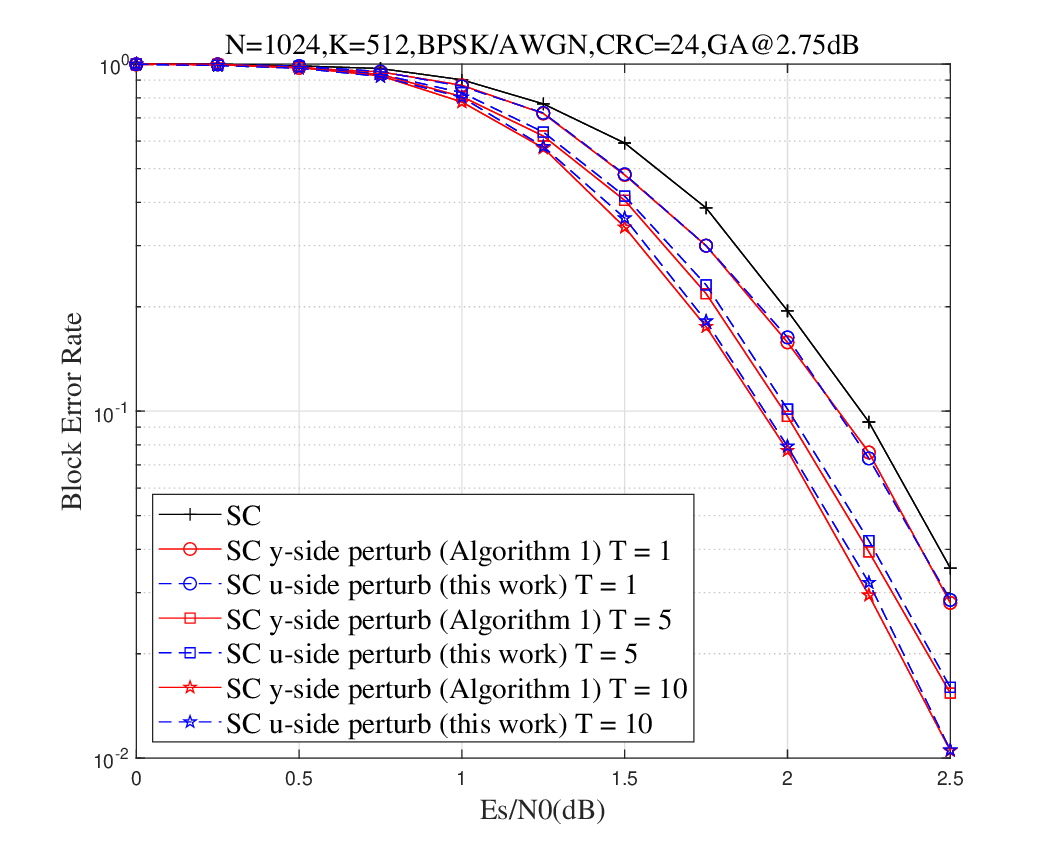}
	}
	\hspace{0.2cm}
	\subfigure[The delay, unchanged, and advance probabilities (\ref{eqi}), (\ref{eqii}), (\ref{eqiii}) for code lengths \( N = \{1024, 2048, 4096, 8192, 16384\} \), code rate \( R = \frac{1}{2} \), and design SNRs \( \{2.75, 2.5, 2.25, 2, 1.85\} \) dB, with a fixed BLER of \( 10^{-2} \), using \( u \)-side perturbation-enhanced SC decoding with \( T = 1 \).]{
		\includegraphics[width=2.9in]{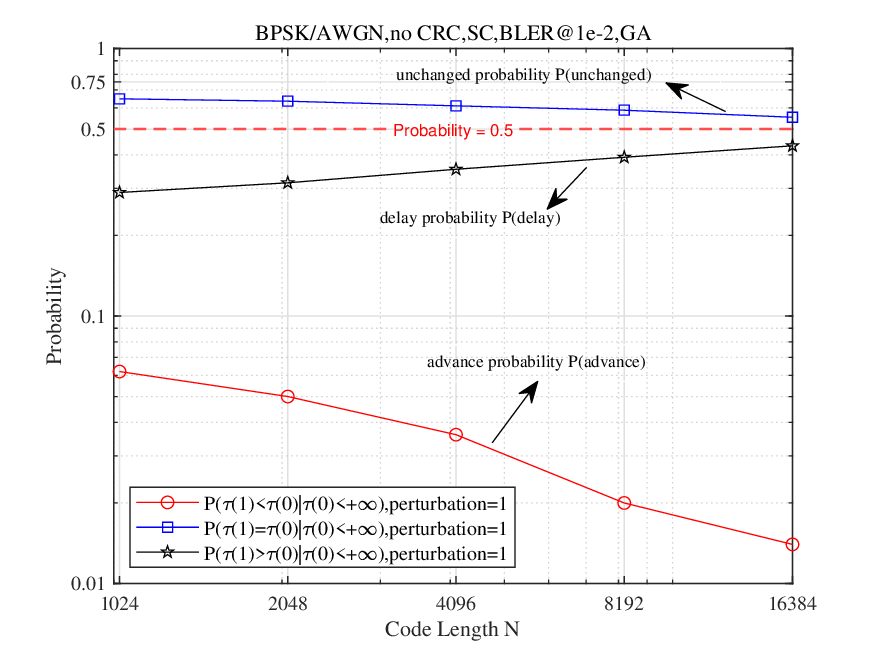}
	}
	\caption{Performance analysis of two perturbation-enhanced SC decoders: $y$-side (Algorithm 1) vs $u$-side (Algorithm 2).}
	\label{fig1}
\end{figure*}
	%the $y$-side perturbation-enhanced SC decoder (Algorithm 1) and $u$-side perturbation-enhanced SC decoder (Algorithm 2)
	All simulations utilize perturbation methods from Algorithm 1 and Algorithm 2 in a BI-AWGN channel with perturbation noise \(\sigma_{p}^{2}=10^{-\frac{\text{SNR}-0.1}{10}}-\sigma^{2}\) \cite{ref13}, where SNR is the signal-to-noise ratio, consistently across all code lengths. Decoding is performed until 400 errors are detected for each code length.
	%All simulations employ perturbation techniques from Algorithm 1 and Algorithm 2 in a BI-AWGN channel with perturbation noise \(\sigma_{p}^{2}=10^{-\frac{\text{SNR}-0.1}{10}}-\sigma^{2}\) \cite{ref14}, where SNR is the signal-to-noise ratio, consistently across all code lengths. Decoding is iterated until 400 errors are detected for each code length.
	%All simulations use perturbation methods from Algorithm 1 and Algorithm 2 under a BI-AWGN channel with perturbation noise \(\sigma_{p}^{2}=10^{-\frac{\text{SNR}-0.1}{10}}-\sigma^{2}\) \cite{ref14}, where SNR is the signal-to-noise ratio, consistent across all code lengths. Decoding continues until 400 errors are detected for each code length.
	%All simulations use the perturbation methods Algorithm 1 and Algorithm 2 under a BI-AWGN channel with perturbation noise \(\sigma_{p}^{2}=10^{-\frac{\text{SNR}-0.1}{10}}-\sigma^{2}\) as advised in \cite{ref14} with SNR the signal-to-noise ratio, consistent across all code lengths. Decoding continues until 400 errors are detected for each code length.
	
	In Fig. \ref{fig1} (a), the red solid line and the blue dashed line represent the performances of the $y$-side (Algorithm 1) and $u$-side (Algorithm 2) perturbation-enhanced SC decoders, respectively. Their nearly identical performance supports further exploration of perturbed decision LLRs in SC decoders.
	%In Fig. \ref{fig1} (a), the red solid and blue dashed lines represent the performances of the $y$-side (Algorithm 1) and $u$-side (Algorithm 2) perturbation-enhanced SC decoders, respectively. They exhibit nearly identical performance, supporting further exploration of perturbed decision LLRs in SC decoders. 
	%the need for additional exploration of perturbed decision LLRs in SC decoders.
	%In Fig. \ref{fig1} (a), the red solid and blue dashed lines show performances for the $y$-side (Algorithm 1) and $u$-side (Algorithm 2) perturbation-enhanced SC decoders, respectively, indicating nearly identical performance and justifying further study of perturbed decision LLRs in SC decoders.
	
	Fig. \ref{fig1} (b) shows that as \(N\) increases, the delay probability of the \(u\)-side perturbation-enhanced SC decoder tends to \(\frac{1}{2}\), and the advance probability decreases to 0, indicating that perturbation either preserves or improves the SC decoder's performance with probability \(\frac{1}{2}\).

	\section{conclusions}
	In this paper, we analyze the delay probability of the first error position and show that adding perturbation does not degrade SC decoder's performance. Our results indicate that the delay probability asymptotically approaches $\frac{1}{2}$, suggesting a potential performance improvement with probability $\frac{1}{2}$. Future research will explore the effects of different perturbation powers on delay probability for a fixed $N$ and develop methods for selecting the optimal perturbation power \(\sigma_p^2\).

	{\appendices

	\end{document}